\documentclass[final]{siamltex}
\usepackage{hhline, emlines, graphicx, epsfig, latexcad, epstopdf}

\title{Limited-Capacity Many-To-Many Point Matching in One Dimension }

\author{Fatemeh Rajabi-Alni
\thanks{Corresponding Author. Department of Computer Engineering, Islamic Azad University, North Tehran Branch, Tehran, Iran, (fatemehrajabialni@yahoo.com).}
\and Alireza Bagheri 
\thanks{Department of Computer Engineering and IT, Amirkabir University of Technology, Tehran, Iran.}
}

\begin{document}

\maketitle

\begin{abstract}
Given two point sets $S$ and $T$, in a \textit{many-to-many matching} between $S$ and $T$ each point in $S$ is assigned to one or more points in $T$ and vice versa. A generalization of the many-to-many matching problem is \textit {the limited capacity many-to-many matching problem}, where the number of points that can be matched to each point, that is the capacity of each point, is limited. In this paper, we provide an $O\left(n^2\right)$ time algorithm for the one dimensional minimum-cost limited capacity many-to-many matching problem, where $\left|S\right|+\left|T\right|=n$. Our algorithm improves the best previous time complexity of $O(kn^2)$, that in which $k$ is the largest capacity of the points in $S \cup T$. In this problem, both $S$ and $T$ lie on the real line and the cost of matching $s \in S$ to $t \in T$ is equal to the distance between $s$ and $t$. 
\end{abstract}



\pagestyle{myheadings}
\thispagestyle{plain}
\markboth{F. RAJABI-ALNI AND A. BAGHERI}{LIMITED-CAPACITY MANY-To-MANY POINT MATCHING IN ONE DIMENSION}

\section{Introduction}
\label{IntroSect}

Given two point sets $S$ and $T$ with $|S|+|T|=n$, a matching function between $S$ and $T$ assigns some points of one set to some points of the other set \cite{1}. The matching has applications in a variety of different fields such as computational biology \cite{2}, operations research \cite{3}, pattern recognition \cite{4}, computer vision \cite{5}. The matching problem has many variations. In a \textit {one-to-one matching} between two point sets, each point of one set is matched to a unique point of the other set, so there exists a perfect correspondence between the members of the sets \cite{6}. In a \textit {many-to-one matching} from $S$ to $T$, each point of $S$ is assigned to one or more points of $T$, and each point of $T$ is assigned to a unique point of $S$ \cite{7}. In a \textit {many-to-many matching} between the $S$ and $T$ each point of $S$ is matched to one or more points of $T$ and vice versa \cite{3}.

Eiter and Mannila \cite{8} reduced the many-to-many matching problem to the minimum-weight perfect matching problem in a bipartite graph, and solved it using the Hungarian method in $O(n^3)$ time. The many-to-many matching between two sets on the real line is solved by reducing it to the problem of finding the shortest path through a directed acyclic graph in $O(n^2)$ time \cite{9,10}, and finally Colannino et al. \cite{1} presented an $O(n \log {n})$-time solution for this problem using the dynamic programming method.

A generalization of the many-to-many matching problem, \textit {limited capacity many-to-many matching problem}, is that in which the number of the points that can be matched to each point is limited. Let $C_S=\{\alpha_1,\alpha_2, \dots,\alpha_y\}$ and $C_T=\{\beta_1,\beta_2,\dots,\beta_z\}$ be the capacity sets of the points in $S$ and $T$, respectively. A minimum-cost limited-capacity many-to-many matching is a matching that matches each point $s_i \in S$ to at least $1$ and at most $\alpha_i$ points in $T$, and each point $t_j \in T$ to at least $1$ and at most $\beta_j$ points in $S$, for all $i,j$ where $1\leq i\leq y$ and $1\leq j\leq z$, such that sum of the matching costs is minimized. Schrijver \cite{11} proved that a minimum-cost limited capacity many-to-many matching can be found in strongly polynomial time.   
A special case of the minimum-cost limited capacity many-to-many matching problem is that in which both $S$ and $T$ lie on the real line and the cost of matching $s_i \in S$ to $t_j \in T$ is equal to the distance between $s_i$ and $t_j$. Panahi and Mohaddes \cite{12} proposed an $O(kn^2)$ time algorithm for the one dimensional minimum limited capacity many-to-many matching, called ODMLM-matching, where $k={\mathop{max}({\mathop{max}_{1\leq i\leq y} {\alpha }_i\ },\ {\mathop{max}_{1\leq j\leq z} {\beta }_j\ })\ }$. In this paper we give an $O(n^2)$ time algorithm for this problem. 

The remainder of this paper is organized as follows. Preliminary definitions are in section \ref{PreliminSect}.  Our algorithm is described in section \ref{OMAsection}. The matching algorithm presented in \cite{1}, computing an optimal many-to-many matching between two point sets on the real line, is briefly described in section \ref{TMMAsection}.

\section{Preliminaries}
\label{PreliminSect}
In this section we describe some preliminary definitions and assumptions used in this paper. We represent both the point $a$ and its $x$-coordinate in the plane using the same symbol $a$. Let $S=\{s_1,s_2,\dots,s_y\}$ and $T=\{t_1,t_2,\dots,t_z\}$ be two sorted sets of points on the real line, that is $s_1<s_2<\dots<s_y$ and $t_1<t_2<\dots<t_z$, such that $s_1$ be the smallest point in $S \cup T$. Let $S \cup T$ be partitioned into subsets $A_0,A_1,A_2,\dots $ such that all points in $A_i$ are smaller than all points in $A_{i+1}$ for all $i$: the point of highest coordinate in $A_i$ lies to the left of the point of lowest coordinate in $A_{i+1}$ (Figure \ref{fig:1}).

Let $A_w=\{a_1,a_2,\dots,a_s\}$ with $a_1< a_2<\dots<a_s$ and $A_{w+1}=\{b_1,b_2,\dots,b_t\}$ with $b_1< b_2<\dots<b_t$. We represent $|b_1-a_i|$ by $e_i$, $|b_i-b_1|$ by $f_i$. It is obvious that $f_1=0$. Moreover, $a_0$ represents the largest point of $A_{w-1}$ for $w>0$. These definitions are presented in Figure \ref{fig:1}. 
In an ODMLM-matching, a point matching to its capacity number of points is called a \textit{saturated} point.

 \begin{figure}
\vspace{-6cm}
\hspace{-13cm}
\resizebox{3\textwidth}{!}{%
  \includegraphics{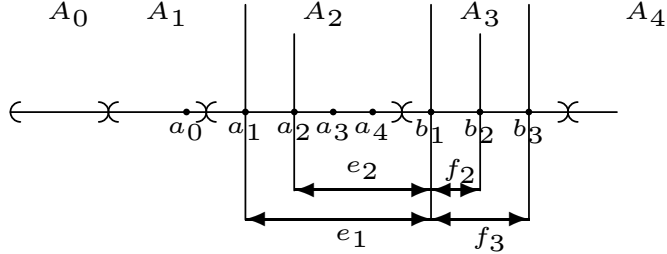}
}
\vspace{-44cm}
\caption{The notation and definitions in partitioned point set $A \cup B$.}
\label{fig:1}       
\end{figure}

\section{The many-to-many matching algorithm in one dimension }
\label{TMMAsection}

In this section, we briefly explain the dynamic programming matching algorithm presented by Colannino et al. \cite{1} which determines a minimum-cost many-to-many matching between $S$ and $T$. The time complexity of their algorithm is $O(n \log n)$ for the unsorted point sets $S$ and $T$, but if the point sets are given in sorted order the algorithm runs in linear time. Let $C(q)$ denote the cost of a minimum cost many-to-many matching for the set of the points $\left\{p\in S\cup T\right|\ p\leq q\}$, the algorithm presented in \cite{1} computes $C(q)$ for all points $q$ in $S\cup T$. In fact, the cost of the minimum many-to-many matching between $S$ and $T$ is equal to $C(m)$, where $m$ is the largest point in $S \cup T$.
In the following, we describe the lemmas and the matching algorithm presented in \cite{1}. The proofs of these lemmas can be found in \cite{1}.

\begin{lemma}
\label{lem1}
Let $a \in S,b\in T$ and $c \in S,d\in T$ such that $a\leq b<c\leq d$. A minimum-cost many-to-many matching contains no pairs $(a,d)$, so any matching $(a,d)$ in a minimum-cost many-to-many matching, with $a<d$, satisfies $a \in A_i$ and $d \in A_{i+1}$, for some $i\geq 0$ (Figure \ref{fig:2}a) \cite{1}. 
\end{lemma}

\begin{lemma}
\label{lem2}
Let $b<c$ be two points in $S$, and $a<d$ be two points in $T$ such that $a\leq b<c\leq d$, then both of $(a,b)$ and $(b,d)$ could not be in an optimal many-to-many matching, so they are mutually exclusive (Figure \ref{fig:2}b) \cite{1}.
\end{lemma}

\begin{lemma}
\label{lem3}
Let $b<c$ be two points in $S$, and $a<d$ be two points in $T$ such that $a\leq b<c\leq d$, then a minimum cost many-to-many matching that contains $(a,c)$ could not contain $(b,d)$, and vice versa (Figure \ref{fig:2}c) \cite{1}.
\end{lemma}

\begin{lemma}
\label{lem4}
In a minimum cost many-to-many matching, each $A_i $ for all $i>0$ contains a point $q_i$, such that all points in $A_i$ lying to the left of $q_i$ are matched with the points in $A_{i-1}$ and all points in $A_i$ lying to the right of $q_i$ are matched with the points in $A_{i+1}$ \cite{1}.
\end{lemma}


 \begin{figure}
\vspace{-5cm}
\hspace{-10cm}
\resizebox{2.5\textwidth}{!}{%
  \includegraphics{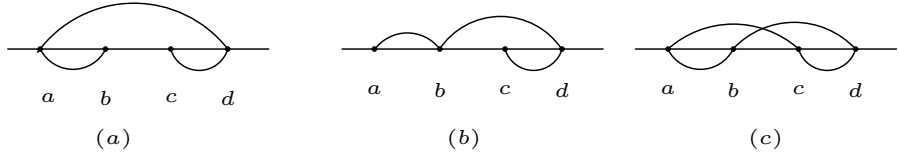}
}
\vspace{-37.5cm}
\caption{Suboptimal matchings. (a) $(a,d)$ is a suboptimal matching. (b) $(a,b)$  and $(b,d)$ do not both belong to an optimal matching. (c) $(a,c)$  and $(b,d)$ do not both belong to an optimal matching.}
\label{fig:2}       
\end{figure}

In fact, we can consider the point $q_i$ defined in Lemma \ref {lem4} as a separating point. As the separating points are explored, their algorithm easily computes the output. They assumed that $C(p)=\infty$ for all points $p \in A_0$. Let $A_w=\{a_1,a_2,\dots,a_s\}$ and $A_{w+1}=\{b_1,b_2,\dots,b_t\}$. Their dynamic programming algorithm's goal is to compute $C(b_i)$ for each $b_i \in A_{w+1}$, assuming that $C(p)$ has been computed for all points $p \leq b_i$ in $S\cup T$. Depending on the values of $w$, $s$, and $i$ five cases arise.

\bf Case 0: \normalfont $w=0$. In this case two situations arise: $i\leq s$ and $i>s$. In the first situation the optimal matching is computed by assigning the first $s-i$ elements of $A_0$ to $b_1$ and the remaining $i$ elements pairwise (Figure \ref{fig:3}a), so $C\left(b_i\right)$ is computed from:
\[C\left(b_i\right)=\sum^s_{j=1}{e_j}+\sum^i_{j=1}{f_j}.\]

In the second situation, $i>s$, the cost is minimized by matching the first $s$ points in $A_1$ pairwise with the points in $A_0$, and the remaining $i-s$ points in $A_1$ with $a_s$ (Figure \ref{fig:3}b). So :
\[C\left(b_i\right)=\left(i-s\right)e_s+\sum^s_{j=1}{e_j}+\sum^i_{j=1}{f_j}.\]


\begin{figure}
\vspace{-5cm}
\hspace{-7cm}
\resizebox{2\textwidth}{!}{%
  \includegraphics{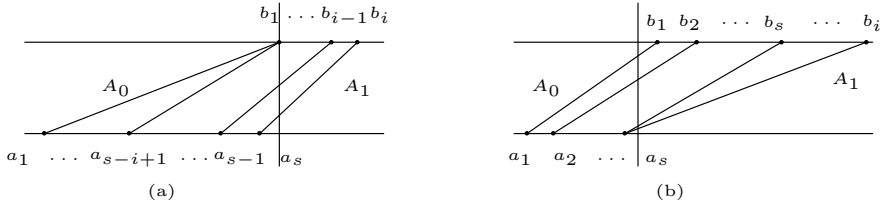}
}
\vspace{-29cm}
\caption{Case 0: $w=0$. (a) $1 \leq i \leq s$. (b) $s<i \leq t$.}
\label{fig:3}       
\end{figure}

\begin{description}

\item[Case 1:] $w>0,s=t=1$. This case is illustrated in Figure \ref{fig:4}a. By Lemma \ref{lem4}, $b_1$ must be matched with the point $a_1$. So we consider $a_1$ only when it reduces the cost of $C\left(b_1\right)$. So:
\[C\left(b_1\right)=e_1+\min(C\left(a_0\right),C\left(a_1\right)).\]

\item[Case 2:] $w>0,s=1,t>1$. By Lemma \ref{lem4}, in a minimum cost matching all points in $A_{w+1}$ should be matched with $a_1$ as presented in Figure \ref{fig:4}b. As case 1, $C\left(b_i\right)$ includes $C\left(a_1\right)$ if $a_1$ covers other points in $A_{w-1}$; otherwise, $C\left(b_i\right)$ includes $C\left(a_0\right)$. So the value of $C\left(b_i\right)$ is:
\[C\left(b_i\right)=\sum^i_{j=1}{f_j}+ie_1+\min\left(C\left(a_0\right),C\left(a_1\right)\right).\] 

\item[Case 3:] $w>0,s>1,t=1$. According to Lemma \ref{lem4}, the point $a_i\in A_w$ should be determined such that all points lying to the left of $a_i$ are matched to points in $A_{w-1}$ and all points lying to the right of $a_i$ are matched to points in $A_{w+1}$ (Figure \ref{fig:4}c). So $a_i$ is the point that minimizes the value of $C\left(b_1\right)$ in the equation: 
\[C\left(b_1\right)={\min_{1\leq i\leq s} (\sum^s_{j=i}{e_j}+C\left(a_{i-1}\right))\ }.\] 


\begin{figure}
\vspace{-4cm}
\hspace{-5cm}
\resizebox{1.8\textwidth}{!}{%
  \includegraphics{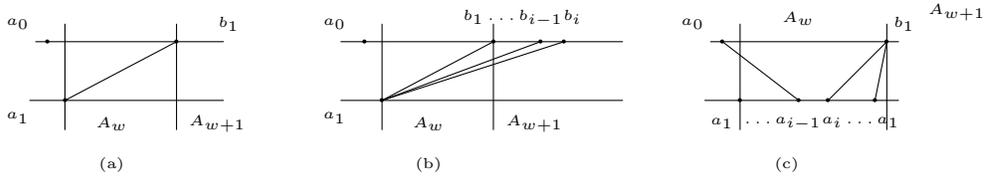}
}
\vspace{-26cm}
\caption{(a). Case 1: $w>0,s=t=1$. (b) Case 2: $w>0,s>1,t=1$. (c) Case 3: $w>0,s>1,t=1$.}
\label{fig:4}       
\end{figure}

\item[Case 4:] $w>0,s>1,t>1$. In this case, we should find the point $q$ that splits $A_w$ to the left and right. Let $X\left(b_i\right)$ be the cost of connecting $b_1,b_2,\dots,b_i$ to at least $i+1$ points in $A_w$ (Figure \ref{fig:5}a). Let $S_i=\sum^s_{j=i}{e_j}+C(a_{i-1})$ for $1\leq i\leq s$ be the cost of connecting the points $a_i,a_{i+1},\dots,a_s$ to $b_1$, plus $C(a_{i-1})$, and $M_i=\min\{S_1,S_2,\dots,S_i\}$, then:
\[X\left(b_i\right)=M_{s-i}+\sum^i_{j=1}{f_j},\ 1\leq i<s.\] 
Let $Y\left(b_i\right)$ be the cost of connecting $b_1,b_2,\dots,b_i$ to exactly $i$ points in $A_w$ (Figure \ref{fig:5}b), then:
\[Y\left(b_i\right)=\sum^s_{j=s-i+1}{e_j}+\sum^i_{j=1}{f_j}+C\left(a_{s-i}\right),\ 1\leq i\leq s. \] 
Finally, let $Z\left(b_i\right)$ denote the cost of connecting $b_1,b_2,\dots,b_i$ to fewer than $i$ points in $A_w$, as depicted in Figure \ref{fig:5}c, then we have:
\[Z\left(b_i\right)={\mathop{min}_{s-i+2\leq j\leq s} (\sum^s_{h=j}{e_h}+\sum^i_{j=1}{f_j}+\left(i+j-s-1\right)e_s+C\left(a_{j-1}\right))\ },\ 1<i.\] 

\end{description}

Note that \[Z\left(b_i\right)=e_s+f_i+min(Y\left(b_{i-1}\right),Z\left(b_{i-1}\right)),\] \[Y\left(b_i\right)=Y\left(b_{i-1}\right)+e_{s-i+1}+f_i+C\left(a_{s-i}\right)-C\left(a_{s-i+1}\right),\] and \[M_i=min(M_{i-1},S_i).\] So the values of $X\left(b_i\right)$, $Y\left(b_i\right)$ and $Z\left(b_i\right)$ can be computed in $O(s+t)$ time. Note that
\[C\left(b_i\right)=\left\{ 
\begin{array}{lr}
\min(X\left(b_i\right),Y\left(b_i\right),Z\left(b_i\right)) & 1\leq i<s 
 \\ 
 \min(Y\left(b_s\right),Z\left(b_s\right)) & i=s 
 \\ 
 C\left(b_{i-1}\right)+e_s+f_i & s<i\leq t 
 \end{array}
\right.,\] 
so the value of $C\left(b_i\right)$ for all $1 \leq i \leq t$ can be computed in $O(s+t)$ time.

\begin{figure}
\vspace{-5cm}
\hspace{-9cm}
\resizebox{2.3\textwidth}{!}{%
  \includegraphics{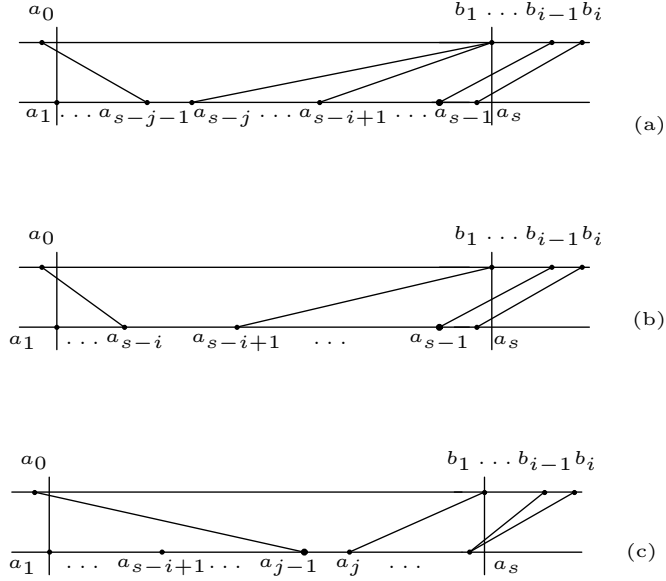}
}
\vspace{-28.5cm}
\caption{Case 4: $w>0,s>1,t>1$. (a) Computing $X(b_i)$. (b) Computing $Y(b_i)$. (c) Computing $Z(b_i)$.}
\label{fig:5}       
\end{figure}

 \section{Our matching algorithm}
 \label{OMAsection}
In this section, we introduce our matching algorithm which is based on the algorithm of Colannino et al. \cite{1}. It is not hard to see that the Lemma \ref{lem3} also applies for an ODMLM-matching. But instead of Lemma \ref{lem1}, Lemma \ref{lem2}, and Lemma \ref{lem4}, we present other lemmas considering the capacity limitations.
Let $Cap(q)$ denote the capacity of the point $q$, i.e. the number of the points that can be matched to $q$. For any point $q$, $C(q,j)$ denotes the cost of an ODMLM-matching for the set of points $\{p\in S\cup T,with \  p\leq q \ and \  Cap\left(q\right)=j\}$. Note that initially $C\left(p_i,k\right)=\infty $ for all $p_i\in S\cup T$ and $1\le k\le Cap(p_i)$. If $m$ is the largest point in $S\cup T$, then $C(m,Cap(m))$ is the cost of an ODMLM-matching.

\begin{lemma}
\label{lem5}
Let $a\in S,b\in T$ and  $c\in S,d\in T$ such that $a\leq b<c\leq d$. If an ODMLM-matching contains the pair $(a,d)$, then either $b$ or $c$ is a saturated point.
\end{lemma}

\begin{proof} 
Suppose that the lemma is false. Let $M$ be an ODMLM-matching that contains such a pair $(a,d)$, and neither $b$ nor $c$ is a saturated point (Figure \ref{fig:2}a). Remove the pair $(a,d)$ from $M$ and add the pairs $(a,b)$ and $(c,d)$. The result $M'$ is still an ODMLM-matching which has a smaller cost, a contradiction.
\qquad\end{proof}

\begin{lemma}
\label{lem6}
Let $b<c$ be two points in $S$, and $a<d$ be two points in $T$ such that $a\le b<c\le d$. If an ODMLM-matching contains both of $(a,b)$ and $(b,d)$, then the point $c$ is a saturated point.  
\end{lemma}

\begin{proof} 
Suppose that the lemma is false. Let $M$ be an ODMLM-matching that contains both $(a,b)$ and $(b,d)$, and $c$ is not a saturated point (Figure \ref{fig:2}b). Then we can remove the pair $(b,d)$ from $M$ and add the pair $(c,d)$: the result $M'$ is still an ODMLM-matching which has a smaller cost, a contradiction.
\qquad\end{proof}

\begin{lemma}
\label{lem7}
Let $a$ be a point in $S$, and $b,c$ be two points in $T$ such that $a\le b<c$. If an ODMLM-matching contains $(a,c)$, then the point $b$ is saturated with some points $d$ in $S$ with $d\le b$. 
\end{lemma}

\begin{proof} 
Let $M$ be an ODMLM-matching that contains $(a,c)$, and the point $b$ is matched to some points $d$ in $S$. Suppose that the lemma is false. For $d$ two cases arise: $b<d<c$, and $c<d$. Assume that $b<d<c$ (Figure \ref{fig6}). Replace $(a,c)$ and $(b,d)$ in $M$ by the two pairs $(a,b)$ and $(d,c)$: the result $M'$ is still an ODMLM-matching that has a smaller cost, a contradiction. Now assume $c<d$, then $M$ contains both of $(a,c)$ and $(b,d)$ (Figure \ref{fig:2}c). This contradicts Lemma \ref{lem3}.
\qquad\end{proof}

\begin{figure}
\vspace{-5cm}
\hspace{-10cm}
\resizebox{2.5\textwidth}{!}{%
  \includegraphics{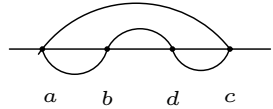}
}
\vspace{-38cm}
\caption{Suboptimal matchings. $(a,c)$  and $(b,d)$ do not both belong to an optimal matching.}
\label{fig6}       
\end{figure}

\begin{lemma}
\label{lem8}
In an ODMLM-matching, consider $A_i$ sets for all $i>0$. Then, each $A_i$ contains a point $q_i$, called here the optimal point, such that all points in $A_i$ less than $q_i$ are matched to some points $b$ with $b\le q_i$, all points in $A_i$ greater than $q_i$ are matched to some points $b'$ with $b'\ge q_i$, and $q_i$ is matched to either some points $b\le q_i$ or some points $b'\ge q_i$.
\end{lemma}

\begin{proof} 
If $A_i$ contains a single point, the lemma is true. Now, assume that $|A_i|>1$ and the lemma is false. Let $a<b$ be two points in $A_i$, the point $a$ be paired with $c$, and the point $b$ be paired with $d$ where $b<c$ and $d<a$. This contradicts Lemma \ref{lem3}.
\qquad\end{proof}

\begin{lemma}
\label{lem9}
Let $A=\{a_1,a_2,\dots,a_s\}$ and $B=\{b_1,b_2,\dots,b_t\}$ be two sets of points on the real line, such that the number of the points of $B$ is equal to sum of the capacities of the points of $A$, i.e. $t=\sum^s_{j=1}{{\alpha }_j}$, and there is an ODMLM-matching between $A$ and $B$. Then $\left|A\right|\le |B|$.
\end{lemma}

\begin{proof} 
Assume that the lemma is not true and  $s>t$, then since the capacity of  each point $a_j\in A$ is ${\alpha}_j\ge 1$, so $t<\sum^s_{j=1}{{\alpha}_j}$. This contradicts the assumption $t=\sum^s_{j=1}{{\alpha}_j}$. 
\qquad\end{proof}


\begin{figure}
\vspace{-2cm}
\hspace{-4.9cm}
\resizebox{1.7\textwidth}{!}{%
  \includegraphics{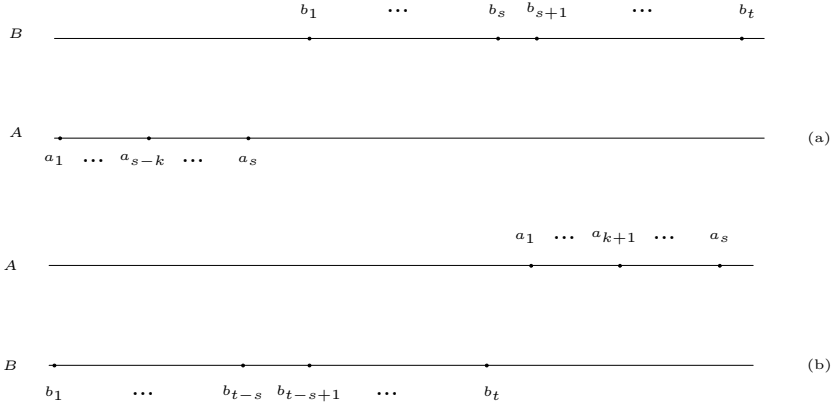}
}
\vspace{-21.5cm}
\caption{Illuustration of Lemma \ref{lem10}. (a) $a_s \le b_1$. (b) $b_t \le a_1$.}
\label{fig:7}       
\end{figure}


\begin{figure}
\vspace{-5cm}
\hspace{-8cm}
\resizebox{2.2\textwidth}{!}{%
  \includegraphics{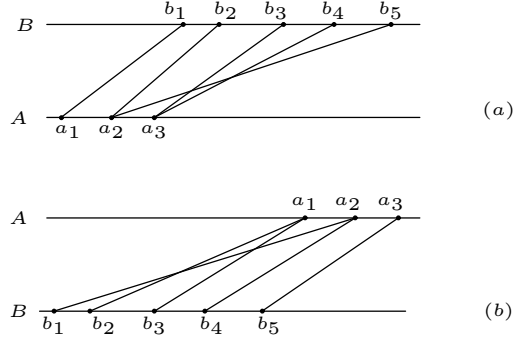}
}
\vspace{-30cm}
\caption{The minimum-cost matching obtained by Lemma \ref{lem10} for $C_A=\{2,2,2\}$: (a) $a_3 \le b_1$. (b) $b_5 \le a_1$.}
\label{fig:9}       
\end{figure}

\begin{lemma}
\label{lem10}
Let $A=\{a_1,a_2,\dots,a_s\}$ and $B=\{b_1,b_2,\dots,b_t\}$ be two sorted sets of limited capacity points on the real line ,such that either $a_s\le b_1$ or $b_t\le a_1$. Let $C_A=\{\alpha_1,\alpha_2,\dots,\alpha_s\}$ and $C_B=\{\beta_1,\beta_2,\dots,\beta_t\}$ be the capacity sets of $A$ and $B$, respectively, such that $t\le \sum^s_{j=1}{{\alpha}_j}$ and $t\ge s$. Then, we can compute an ODMLM-matching between $A$ and $B$ in $O(s+t)$ time. 
\end{lemma}

\begin{proof} 
Either $a_s\le b_1$ or $a_1\ge b_t$. Obviously, the minimum cost is obtained by matching each point in $B$ with as closets point as possible in $A$. So, the two sets are matched as following.

Assume $t \le\sum^s_{j=1}{{\alpha}_j}$ and $a_s\le b_1$ (Figure \ref{fig:7}a). Initially, we match the first $s$ points in $B$ pairwise with the points in $A$. This guarantee that each point of $A$ is matched at least once. Then, starting from $b_{s+1}$, we match each point of $B$ to the closets point of $A$, here the largest point of $A$, that has not been saturated yet. So the points in $A$ are matched with the points in $B$ such that the last $k$ points of $A$ are saturated, the first $s-k-1$ points of $A$ are single matched, and the $\left(s-k\right)th$ point of $A$ may be neither saturated nor single matched, where $k$ is the greatest integer that $s-k-1+\sum^s_{j=s-k+1}{\alpha_j}<t$, i.e. $\sum^s_{j=s-k+1}{\alpha_j}<t-s+k+1$. In fact, $a_{s-k}$ is the border between the single matched points in $A$ and the saturated points, such that all points in $A$ lying to the left of $a_{s-k}$ are single matched and all points in $A$ lying to the right of it are saturated. Figure \ref{fig:9}a illustrates an example of this situation.

Now assume $t\le\sum^s_{j=1}{{\alpha}_j}$ and $a_1\ge b_t$. In this situation, that is shown in Figure \ref{fig:8},we assign the last $s$ elements in $B$ pairwise with the points in $A$ to guarantee that each point of $A$ is matched to at least one point of $B$. Then starting from $b_{t-s}$, we match each point of the remaining $t-s$ points in $B$, that is $b_1, \dots , b_{t-s}$, to the smallest point in $A$ that has not been saturated. Hence, the first $k$ points in $A$ are saturated, the last $s-k-1$ points are single matched, and the $\left(k+1\right)th$ point may be neither single matched nor saturated where $k$ is the greatest integer that $s-k-1+\sum^k_{j=1}{\alpha_j}<t$, i.e. $\sum^k_{j=1}{\alpha_j}<t-s+k+1$. Figure \ref{fig:7}b illustrates an example of this case of Lemma \ref{lem10}. 
\qquad\end{proof}

\begin{theorem}
Let $S$ and $T$ be two sets of points on the real line with $\left|S\right|+\left|T\right|=n.$ Then, a minimum-cost many-to-many matching with limited capacity between $S$ and $T$ can be determined in $O(n^2)$ time.
\end{theorem}

\begin{proof} We can compute an ODMLM-matching, by finding the optimal points defined in Lemma \ref{lem8}. The optimal point of each set $A_i$ for all $i$ separates the points lying to the left of the optimal point from the points lying to the right of it. Consider $A_w=\{a_1,a_2,\dots,a_s\}$ and $A_{w+1}=\{b_1,b_2,\dots,b_t\}$. Let $C_w=\{{\alpha}_1,{\alpha}_2,\dots,{\alpha}_s\}$ and $C_{w+1}=\{{\beta}_1,{\beta}_2,\dots,{\beta}_t\}$ be the capacity sets of the points in $A_w$ and $A_{w+1}$, respectively . 

Assume that we have computed $C\left(p,k\right)$ for all points $\{p\in S\cup T,\ with\ p\le a_s\ and\ 1\le k\le Cap(p)\}$, and now we want to compute $C(b_i,k)$ for each $b_i\in A_{w+1}and\ 1\le k\le {\beta }_i$. Our algorithm consists of two steps, the primary step and the main step explained in the following.

\textbf{Primary step.} starting from $a_1$, we investigate each $a_j\in A_w$ one by one in the ascending order. For each $a_j\in A_w$ that the demands of the points $a_j,a_{j+1},\dots,a_{s}$ can not be satisfied by the points in $A_{w+1}$, i.e. $\sum^t_{k=1}{{\beta}_k}<\left(s-j+1\right)$, we seek the next sets to find enough capacities according to Lemma \ref{lem7}. Let $C=\sum^t_{j=1}{\beta_j}$, and $A=\{a_j,a_{j+1},\dots,a_{j+C-1}\}$ be the $C$ consecutive elements of $A_w$ starting from $a_j$ (Figure \ref{fig:10}). We match $a_j,a_{j+1},\dots,a_{j+C-1}$ with the points of $A_{w+1}$ using Lemma \ref{lem10}, so the last $\beta_1$ points in $A$ are matched to the point $b_1$, the last $\beta_2$ points of the remaining $C-\beta_1$ points in $A$ are matched to the point $b_2$, and so on. Let $C'$ be the cost of matching the points in $A$ to the points in $A_{w+1}$ using Lemma \ref{lem10}, then:
\[C'=\sum^{j+C-1}_{k=j}{e_k}+\sum^t_{k=1}{\beta_kf_k}+min(C\left(a_{j-1},{\alpha}_{j-1}\right),C\left(a_j,{\alpha }_j-1\right)). \] 
Let \[A=\{a_{j+C},a_{j+C+1},\dots,a_s\}\cup A_{w+2}=\{{a'}_1,{a'}_2,\dots,{a'}_{s'}\},\] and
$A_{w+3}=\{{b'}_1,{b'}_2,\dots,{b'}_{t'}\}$. Let $C_{w+3}=\{{\beta}'_1,{{\beta}'}_2,\dots,{{\beta}'}_{t'}\}$ be the capacities of the points of $A_{w+3}$, and $C=\sum^{t'}_{j=1}{{{\beta}'}_j}$. Let ${e'}_i=|{b'}_1-{a'}_i|$, and ${f'}_i=|{b'}_i-{b'}_1|$. 
If $A_{w+3}$ can satisfy the demands of the points in $\{{a'}_1,{a'}_2,\dots,{a'}_{s'}\}$, that is $C\ge s'$, the process is stopped. Otherwise, if $C<s'$, we match the first $C$ points in $A$, that is $A'=\{{a'}_1,{a'}_2,\dots,{a'}_C\}$, with the points in $A_{w+3}$ using Lemma \ref{lem10}. So the last ${{\beta}'}_1$ points of $A'$ are matched to the point $b_1$, the last ${\beta '}_2$ points of the remaining $C-{{\beta}'}_1$ points in $A'$ are matched to the point $b_2$, and so on. Then, let
\[C'=C'+\sum^C_{j=1}{{e'}_j}+\sum^{t'}_{j=1}{{{\beta}'}_j{f'}_j}. \] 
Then, we compare sum of the capacities of the points in $A_{w+5}$ with the number of the points in $\{{a'}_{C+1},{a'}_{C+2},\dots,{a'}_{s'}\}\cup A_{w+4}$, so forth. 

This forward process is followed to find the first set, called $A_{j'}$, which can satisfy the demands of the remaining unmatched points in $$A=\left\{a_j,a_{j+1},\dots,a_s\right\}\cup \bigcup^{j'-1}_{k=w+2}{A_k}.$$ Note that it is possible that there does not exist a set such $A_{j'}$. If $A_{j'}$ exists, we should match the unmatched points in $A$ with the points in $A_{j'}$. 

In fact, by Lemma \ref{lem7} the points in $A_k$ can be matched to the points in $A_{k'}$ with $k<k'$ for all $w \le k \le j'-1$ (Figure \ref{fig:11}). Figure \ref{fig:12} shows an example for this situation.


\begin{figure}
\vspace{-4cm}
\hspace{-4cm}
\resizebox{1.5\textwidth}{!}{%
  \includegraphics{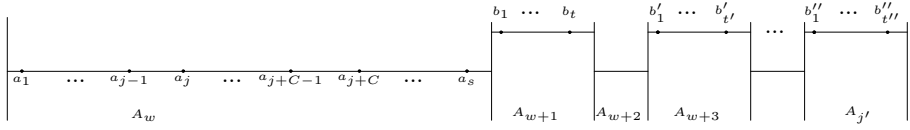}
}
\vspace{-22cm}
\caption {A partitioned point set that is used for illustration of the situation where $\sum^t_{k=1}{{\beta}_k}<\left(s-j+1\right)$.}
\label{fig:10}       
\end{figure}


\begin{figure}
\vspace{-5.3cm}
\hspace{-10cm}
\resizebox{2.6\textwidth}{!}{%
  \includegraphics{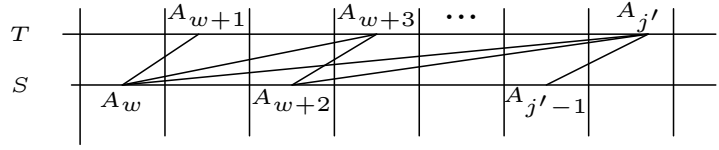}
}
\vspace{-39cm}
\caption{Illustration of the situation $\sum^t_{k=1}{{\beta}_k}<\left(s-j+1\right)$.The sets that can be matched in an ODMLM-matching are connected with a line.}
\label{fig:11}       
\end{figure}


\begin{figure}
\vspace{-4cm}
\hspace{-6cm}
\resizebox{1.8\textwidth}{!}{%
  \includegraphics{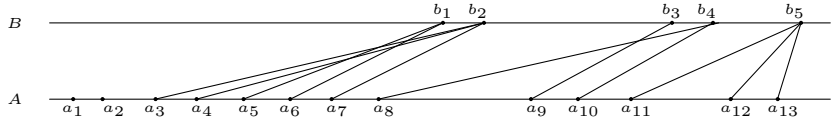}
}
\vspace{-27cm}
\caption{An example for the situation $\sum^t_{k=1}{{\beta}_k}<\left(s-j+1\right)$ with $C_B=\{2, 3, 1, 2, 4\}$.}
\label{fig:12}       
\end{figure}

Let \[A=\{{a''}_1,{a''}_2,\dots,{a''}_{s''}\}=\left\{a_j,a_{j+1},\dots,a_s\right\}\cup \bigcup^{j'-1}_{k=w+2}{A_k},\] and $A_{j'}=\{{b''}_1,{b''}_2,\dots,{b''}_{t''}\}$. Consider $C_A=\{{{\alpha}''}_1,{\alpha ''}_2,\dots,{\alpha ''}_{s''}\}$ and $C_{j'}=\{{\beta ''}_1,{\beta ''}_2,\dots,{\beta ''}_{t''}\}$ as the capacity sets of $A$ and $A_{j'}$, respectively. Let ${e''}_i=|{b''}_1-{a''}_i|$ and ${f''}_i=|{b''}_i-{b''}_1|$. 

Depending on the values of $s''$ and $t''$ four cases arise:
\begin{enumerate}
\item $s''=1,t''=1$. In this case ${a''}_1$ is paired with ${b''}_1$. Let $C'=C'+{e''}_1$, then we have:
\[C\left({b''}_1,k\right)={\min(C\left({b''}_1,k\right),\ C')\ } \ for \  all \ 1\le k\le {\beta ''}_1.\]

\item $s''=1,t''>1$. In this situation we compute the cost of matching the points ${b''}_1,{b''}_2,..,{b''}_{i'}$ to ${a''}_1$, and then update the costs of $C\left({b''}_{i'},k\right)$ for all $1 \le i'\le t''$:
\[C\left({b''}_{i'},k\right)={\min(C\left({b''}_{i'},k\right),\ C'+\sum^{i'}_{j=1}{{f''}_j}+i'{e''}_1)\ }\  for \ all \ 1\le k\le {\beta ''}_{i'}.\]

\item  $s''>1,t''=1$. In this situation the points ${a''}_1,{a''}_2,\dots,{a''}_{s''}$ are matched with ${b''}_1$, let $C'=C'+\sum^{s''}_{j=1}{{e''}_j}$, then:
\[C\left({b''}_1,k\right)=\ {\min(C\left({b''}_1,k\right),\ C'+\sum^{s''}_{j=1}{{e''}_j})\ } \ for \ all \ s''\le k\le {\beta ''}_1.\]

\item  $s''>1,t''>1$. In this case we compute the cost of matching ${a''}_1,{a''}_2,\dots,{a''}_{s''}$ to ${b''}_1,{b''}_2,\dots,{b''}_{i'}$ for all $1 \le i' \le t''$ using Lemma \ref{lem10}. We discuss two cases; $s''\le i'$ or  $s''>i'$. 

First assume $s''\le i'$. In this case either $\sum^{s''}_{j=1}{{\alpha ''}_j}=i'$ or $\sum^{s''}_{j=1}{{\alpha ''}_j}>i'$. If $\sum^{s''}_{j=1}{{\alpha ''}_j}=i'$, then by Lemma \ref{lem10} we have:
\[C'=C'+\sum^{s''}_{j=1}{{{\alpha ''}_je''}_j}+\sum^{i'}_{j=1}{{f''}_j}.\] 
And if $\sum^{s''}_{j=1}{{\alpha ''}_j}>i'$, then let $k'$ be the greatest integer that \[s''-k'-1+\sum^{s''}_{j=s''-k'+1}{{\alpha ''}_j}<i',\]
and \[m=i'-\left(s''-k'-1+\sum^{s''}_{j=s''-k'+1}{{\alpha ''}_j}\right)=i'-s''+k'+1-\sum^{s''}_{j=s''-k'+1}{{\alpha ''}_j},\] and \[C'=C'+\sum^{s''}_{j=s''-k'+1}{{\alpha ''}_j{e''}_j}+m{e''}_{s''-k'}+\sum^{s''-k'-1}_{j=1}{{e''}_j}+\sum^{i'}_{j=1}{{f''}_j},\] according to Lemma \ref{lem10}. Then, for all $s'' \le i' \le t''$ :
\[C\left({b''}_{i'},k\right)={min (C\left({b''}_{i'},k\right),\ C')\ }\  for \ all\  1\le k\le {\beta ''}_{i'}.\]
 Now assume $s''>i'$. In this situation either $\sum^{i'}_{j=1}{{\beta ''}_j}=s''$ or $\sum^{i'}_{j=1}{{\beta ''}_j}>s''$. If $\sum^{i'}_{j=1}{{\beta ''}_j}=s''$, then:
\[C\left({b''}_{i'},{\beta ''}_{i'}\right)={min (C\left({b''}_{i'},{\beta ''}_{i'}\right),\ C'+\sum^{s{''}}_{j=1}{{e''}_j}+\sum^{i'}_{j=1}{{\beta ''}_j{f''}_j})\ }.\] 
And if $\sum^{i'}_{j=1}{{\beta ''}_j}>s''$, then let $k'$ be the greatest integer that \[i'-k'-1+\sum^{k'}_{j=1}{{\beta ''}_j}<s'',\] \[m=s''-\left(i'-k'-1+\sum^{k'}_{j=1}{{\beta ''}_j}\right)=s''-i'+k'+1-\sum^{k'}_{j=1}{{\beta ''}_j},\] and \[C'=C'+\sum^{s''}_{j=1}{{e''}_j}+\sum^{k'}_{j=1}{{\beta ''}_j{f''}_j}+m{f''}_{k'+1}+\sum^{i'}_{j=k'+2}{{f''}_j}.\] 
Then, $C\left({b''}_{i'},k\right)={\min(C\left({b''}_{i'},k\right),\ C')\ }$, 
\[for\ \left\{ 
 \begin{array}{ll}
k={\beta ''}_{i'} & if\ i'=k' 
 \\ 

m\le k\le {\beta ''}_{i'} & if\ i'=k'+1
 \\ 
 
1\le k\le {\beta ''}_{i'} & if\ k'+1<i' 
 \end{array}
\right.\] 

\end{enumerate}

\textbf{Main step.} Now we start to compute $C\left(b_i,k\right)$ for all $1\le i\le t$ and $1\le k\le \beta_i$. For computing $C(b_i,k)$ for each $b_i\in A_{w+1}and\ 1\le k\le {\beta }_i$ two general cases arise:
\begin{description}
\item[Case A:] $i\le \sum^s_{j=1}{{\alpha }_j}$. In this case where the demands of the points in $A_{w+1}$ can be satisfied with the points in $A_w$, we discuss five cases. 

\begin{description}

\item[Case A.0:] $w=0$. This case is similar to Case 0 of [1]. If $i=\sum^s_{j=1}{{\alpha}_j}$, then by Lemma \ref{lem10} for all $1\le j\le {\beta}_i$ we have
\[C\left(b_i,j\right)=\sum^s_{k=1}{{\alpha}_ke_k}+\sum^i_{k=1}{f_k}. \] 
For $i<\sum^s_{j=1}{{\alpha}_j}$ we discuss two cases: $s\le i$, or $s>i$. Assume first that $s\le i$. In this situation, we match the points in $A_0$ with the points in $A_1$ by Lemma \ref{lem10}. So the first $s$ points in $A_1$, that is $b_1,b_2,\dots,b_s$, are matched pairwise with the points in $A_0$. 

If $s<i$, starting from $b_{s+1}$ each point of the remaining $i-s$ points in $A_1$ is matched to the largest point in $A_0$ that is not saturated.

Hence, the last $k$ points of $A_0$ are saturated, the first $s-k-1$ points are single matched, and $\left(s-k\right)th$ point of $A_0$ may be neither saturated nor single matched, where $k$ is the greatest integer that $s-k-1+\sum^s_{l=s-k+1}{{\alpha}_l}<i$ or $\sum^s_{l=s-k+1}{{\alpha}_l}<i-s+k+1$. 
Let \[m=i-\left(s-k-1+\sum^s_{l=s-k+1}{{\alpha}_l}\right)=i-s+k+1-\sum^s_{l=s-k+1}{{\alpha}_l}\] be the number of the points of $A_1$ that are matched to $a_{s-k}$. Then the value of $C\left(b_i,j\right)$ for all $1\le j\le {\beta}_i$ is:
\[C\left(b_i,j\right)=\sum^s_{l=s-k+1}{{\alpha}_le_l}+me_{s-k}+\sum^{s-k-1}_{l=1}{e_l}+\sum^i_{l=1}{f_l}\] 

Now assume that $s>i$. We identify three cases in this situation:

\begin{enumerate}

\item  $\sum^i_{j=1}{{\beta}_j}<s$. In this situation, we define $C\left(b_i,j\right)=\infty $ for all $1\le j\le {\beta}_i$.

\item  $\sum^i_{j=1}{{\beta}_j}=s$. In this situation by Lemma \ref{lem10}, the cost of the matching is minimized when the last ${\beta}_1$ points of $A_0$ are matched to the point $b_1$, the last ${\beta}_2$ points of the remaining $s-{\beta}_1$ points in $A_0$ are matched to the point $b_2$, and so on. Let $R'=\sum^s_{j=1}{e_j}+\sum^i_{j=1}{{\beta}_jf_j}$, then:
\[C\left(b_i,j\right)=\left\{ \begin{array}{cc}
\infty  & 1\le j<{\beta }_i \\ 
R' & j={\beta }_i \end{array}
\right.. \] 

\item  $\sum^i_{j=1}{{\beta}_j}>s$. By Lemma \ref{lem10}, in this situation the optimal matching is computed by assigning the last $i$ elements in $A_0$, $a_{s-i+1}, \dots , a_s$, pairwise with the points in $A_1$. Then, starting from $a_{s-i}$ we match each point of the remaining unmatched points in $A_0$ to the closest point, the smallest point, of $A_1$ that has not been saturated. Let \[R=\sum^s_{j=1}{e_j}+\sum^k_{j=1}{{\beta}_jf_j}+mf_{k+1}+\sum^i_{j=k+2}{f_j},\] $k$ be the greatest integer that \[i-k-1+\sum^k_{j=1}{{\beta}_j}<s,\] and \[m=s-\left(i-k-1+\sum^k_{j=1}{{\beta}_j}\right)=s-i+k+1-\sum^k_{j=1}{{\beta }_j}.\] 
\end{enumerate}

Then, $C\left(b_i,j\right)=R$,
\[for\ \left\{ 
 \begin{array}{ll} 
j={\beta }_i & if\ i=k 
 \\ 

m\le j\le {\beta }_i & if\ i=k+1 
 \\ 
 
1\le j\le {\beta }_i & if\ k+1<i 
 \end{array}
\right..\]

\item[Case A.1:] $w>0$, $\beta_1=1$, $i=1$. In this case, Lemma \ref{lem8} implies that $b_1$ must be paired with $a_s$. As Case 1 of [1], we identify two case: $(i)$ $a_s$ is paired with both $b_1$ and some other points $p<b_1$, and $(ii)$ $a_s$ is paired with only $b_1$. We choose the matching of minimum cost:

\[C\left(b_1,k\right)=\min(C\left(b_1,k\right),e_1+{\min \left(C\left(a_{s-1},{\alpha}_{s-1}\right),C\left(a_s,{\alpha}_s-1\right)\right)\ })\] for all $1\le k\le {\beta}_1$.

\item[Case A.2:] $w>0$, $s=1$, $i>1$. By Lemma \ref{lem8} all points in $A_{w+1}$ are assigned to $a_s$. As Case A.1, other points $p<b_1$ may be assigned to $a_s$. So:

\[C\left(b_i,k\right)=\min(C\left(b_i,k\right),\sum^i_{j=1}{f_j}+ie_1+{\min \left(C\left(a_{0},{\alpha}_{0}\right),C\left(a_1,{\alpha}_1-i\right)\right)})\] for all $1\le k\le {\beta }_i$.

\item[Case A.3:] $w>0$, $s>1$, $\beta_1>1$, $i=1$. In this case we should examine all points ${a}_{j}\in {A}_{w}$ and determine the point $q_i$ defined in Lemma \ref{lem8}. Let 

\[C_k={\min_{{s-k+1}\le h\le s} (\sum^s_{i=h}{e_i}+min(C\left(a_{h-1},\alpha_{h-1}\right),C\left(a_h,{\alpha_h-1}\right))},\] then we have: 

\[C\left(b_1,k\right)=min(C\left(b_1,k\right),C_k) \ for all \ 1\le k \le \beta_1.\]

\item[Case A.4:] $w>0$, $s>1$, $i>1$. In this situation, by Lemma \ref{lem8} we should find the optimal point of $A_w$. 

Let $Y(b_i)$ and $Z(b_i)$ represent the cost of connecting $b_1,b_2,\dots,b_i$ to exactly $i$ points, and fewer than $i$ points in $A_w$, respectively. Let $X(b_i,j)$ represents the cost of connecting $b_1,b_2,\dots,b_i$ to at least $i+1$ points, where the number of the points that can be matched to $b_i$ is limited to $j$, with $1\le j \le \beta_i$. The values of $X(b_i,j)$, $Y(b_i)$ and $Z(b_i)$ are computed by Lemma \ref{lem10} as following.

Let $S'_i$ be the cost of connecting the points $a_i,\ a_{i+1},\ \dots ,\ a_s$ to $b_1$ plus $min\left(C\left(a_{i-1},{\alpha }_{i-1}\right),C\left(a_i,{\alpha }_i-1\right)\right)$ for $1\le i\le s$, so  $$S'_i=\sum^s_{j=i}{e_j}+min\left(C\left(a_{i-1},{\alpha }_{i-1}\right),C\left(a_i,{\alpha }_i-1\right)\right).$$ 

Let $${S'}_{i,j}=S'_i+\sum^k_{l=1}{{\beta }_lf_l}+mf_{k+1}+\sum^j_{l=k+2}{f_l},$$ be the cost of matching $a_i, a_{i+1}, \dots, a_s$ to $b_1, b_2, \dots, b_j$ with $j<s-i+1$, plus $min\left(C\left(a_{i-1},{\alpha }_{i-1}\right),C\left(a_i,{\alpha }_i-1\right)\right)$, where $k$ is the greatest integer that $$j-k-1+\sum^k_{l=1}{{\beta }_lf_l}<s-i+1,$$ and \[m=s-i+1-j+k+1-\sum^k_{l=1}{{\beta }_lf_l}.\] 

Let ${M'}_{i,j}=min\{{S'}_{1,j},\ \dots ,{S'}_{i-1,j},\ {S'}_{i,j}\}$. Note that if $$\sum^{j-1}_{l=1}{{\beta }_l}+k<s-i+1$$ then \[{S'}_{i,j}=\infty \  for \  1\le k\le {\beta }_j.\]

Now we compute the values of $X\left(b_i,j\right)$, $Y\left(b_i\right)$, and $Z\left(b_i\right)$ for $1\le i\le min(s,t)$ and $1\le j\le {\beta }_i$ as following:
\[X\left(b_i,j\right)={M'}_{s-i,i},\] 
and 
\[Y\left(b_i\right)=\sum^{s-i+1}_{j=s}{e_j}+\sum^i_{j=1}{f_j}+min(C\left(a_{s-i},{\alpha }_{s-i}\right),C\left(a_{s-i+1},{\alpha }_{s-i+1}-1\right)).\] For the value of $Z\left(b_i\right)$ we have:
\[Z\left(b_i\right)={\mathop{min}_{s-i+2\le h\le s} \left(R_{ih}\right)\ },\] 

where \[R_{ih}=\left\{ \begin{array}{ll}
R+C\left(a_{h-1},{\alpha }_{h-1}\right)& if\ k=h\\ 
R+min(C\left(a_{h-1},{\alpha }_{h-1}\right),C\left(a_h,{\alpha }_h-m\right)& if\ k-1=h\\ 
R+min(C\left(a_{h-1},{\alpha }_{h-1}\right),C\left(a_h,{\alpha }_h-1\right) & if\ \ k-1>h\end{array}\right.\] 
That in which $R$ is the cost of connecting $a_h, \dots, a_s$ to $b_1, \dots, b_i$, so $$R=\sum^s_{j=s-k+1}{{\alpha }_je_j}+me_{s-k}+\sum^{s-k-1}_{j=h}{e_j}+\sum^i_{j=1}{f_j},$$
where $k$ is the greatest integer for which $$s-h+1-k-1+\sum^{s}_{j=s-k+1}{{\alpha }_j}<i,$$  $$m=i-s+h+k-\sum^{s}_{j=s-k+1}{{\alpha }_j}.$$ 
In fact, $k=h$ implies that $a_h$ is a saturated point, $k-1=h$ implies that $m$ points are matched to $a_h$, and $k-1>h$ implies that $a_h$ is matched to a single point.

Let $a_j$ be the smallest point in $A_w$ that for which $$\sum^{i-1}_{l=1}{{\beta }_l}+k\ge s-j+1,$$ then the value of $C\left(b_i,k\right)$ for all $1\le k\le {\beta }_i$ is:
\[C\left(b_i,k\right)=\left\{ 
 \begin{array}{ll}
{min \left(X\left(b_i,k\right),\ Y\left(b_i\right),\ Z\left(b_i\right),C\left(b_i,k\right)\right)\ } & i<s-j+1 
 \\ 
 
{min \left(Y\left(b_i\right),\ Z\left(b_i\right),C\left(b_i,k\right)\right)\ } & i=s-j+1 
 \\ 
 
min({\mathop{min}_{j\le h\le s\ } \left(R_{ih}\right)\ },C\left(b_i,k\right)) & i>s-j+1 
 \end{array}
\right..\]

\end{description}

\item[Case B:] $w>0$, $i>\sum^s_{j=1}{{\alpha }_j}$. In this case, the sum of the capacities of the points in $A_w$ is less than the number of the points in $A_{w+1}$, so by Lemma \ref{lem7} we should seek the previous sets to find enough capacities for $b_1,b_2,\dots,b_i$ (Figure \ref{fig:13}). Let $C=\sum^s_{j=1}{{\alpha }_j}$, then the last $C$ elements in $A_{w+1}$, that is $b_{i-C+1},\dots,b_{i-1},b_i$, are matched to the points in $A_w$ according to Lemma \ref{lem10}. The cost of this matching is:
\[C'=\sum^s_{j=1}{{{\alpha}_je}_j}+\sum^i_{j=i-C+1}{f_j}.\] 
Let \[A_{w-1}\cup \left\{b_1,b_2,\dots,b_{i-C}\right\}=\{{b'}_1,{b'}_2,\dots,{b'}_{t'}\}\ with \ {b'}_1<{b'}_2<\ \dots <{b'}_{t'},\] \[A_{w-2}=\{{a'}_1,{a'}_2,\dots ,{a'}_{s'}\}\ with\ {a'}_1<{a'}_2<\ \dots <{a'}_{s'},\] \[C_{w-2}=\{{\alpha '}_1,{\alpha '}_2,\dots,{\alpha '}_{s'}\},\] and $C'=\sum^{s'}_{j=1}{{\alpha '}_j}$. 

If sum of the capacities of the points in $A_{w-2}$ is greater than or equal to the number of the points in $A_{w-1}\cup \{b_1,b_2,\dots,b_{i-C}\}$, that is $t'\le C'$, the process is stopped. Otherwise, the last $C'$ elements in $\{{b'}_1,{b'}_2,\dots ,{b'}_{t'}\}$, that is $\{{b'}_{t'-C'+1},\dots,{b'}_{t'-1},{b'}_{t'}\}$, are matched to the points in $A_{w-2}$ according to Lemma \ref{lem10}. Let ${e'}_i=|{b'}_1-{a'}_i|$ and ${f'}_i=|{b'}_i-{b'}_1|$, then:
\[C'=C'+\sum^{s'}_{j=1}{{{\alpha '}_je'}_j}+\sum^{t'}_{j=t'-C'+1}{{f'}_j}.\] 
Now we compare the sum of the capacities of the points in $A_{w-4}$ with the number of the points in  $A_{w-3}\ \cup \left\{{b'}_1,{b'}_2,\ \dots ,{b'}_{t'-C'}\right\}$, and so on.


\begin{figure}
\vspace{-3cm}
\hspace{-3.5cm}
\resizebox{1.5\textwidth}{!}{%
  \includegraphics{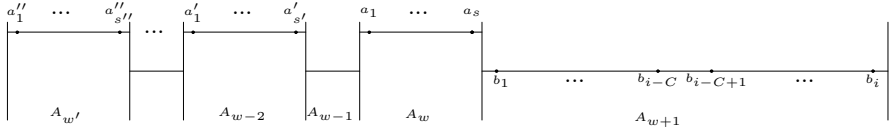}
}
\vspace{-22cm}
\caption{A partitioned point set used for illustration of Case B: $i>\sum^s_{j=1}{{\alpha }_j}$.}
\label{fig:13}       
\end{figure}

This backward process is followed until finding the first set, called $A_{w'}$, that can satisfy the demands of the remaining unmatched points in the set \[A_{w'+1}\cup \dots \ \cup \ A_{w-1}\ \cup \ A_{w+1},\] i.e. the sum of the capacities of the points in $A_{w'}$ is greater than or equal to the number of the unmatched points in the set $A_{w'+1}\cup \dots \ \cup \ A_{w-1}\ \cup \ A_{w+1}$. It is possible that we do not reach such a set $A_{w'}$, in this situation $C\left(b_i,k\right)=\infty $ for all $1\le k\le {\beta}_i$.

Let $A=\{{b''}_1,{b''}_2,\dots,{b''}_{i'}\}$ be the set of the remaining points in $\bigcup^{w+1}_{k=w'+1}{A_k}$, $A_{w'}=\{{a''}_1,{a''}_2,\dots,{a''}_{s''}\}$, $a_0$ the largest point in $A_{w'-1}$. 

Let $C_A=\{{\alpha ''}_1,{\alpha ''}_2,\dots,{\alpha ''}_{s''}\}$ and $C_B=\{{\beta ''}_1,{\beta ''}_2,\dots,{\beta ''}_{i'}\}$ be the capacity sets of the points in $A_{w'}$ and $A$, respectively. Let $C'\left({b''}_{i'},k'\right)$ be the cost of matching the points in $A_{w'}$ and $A$ using Case A.0, Case A.1, Case A.2, Case A.3, or Case A.4, depending on the values of $w'$, $s''$, and $i'$. Then, we have

\[C\left(b_i,k\right)={\mathop{\min}_{1\le k'\le {{\beta }''\ }_{i'}} C'\left({b''}_{i'},k'\right)\ }+C' \  for\  all\  1\le k\le {\beta}_i.\]

Figure \ref{fig:14} illustrates an example of Case B, where $A_0=\{s_1\}$, $A_1=\{t_1\}$, $A_2=\{s_2,s_3\}$, $A_3=\{t_2,t_3,t_4,t_5,t_6,t_7\}$, $C_S=\{5,2,3\}$, and $C_T=\{2,3,3,1,2,3,3\}$. First, the last $5$ points in $A_3$ are mapped to the points in $A_2$, then we seek the previous sets to find enough capacities, so we match $t_2$ with $s_1$.


\begin{figure}
\vspace{-5.3cm}
\hspace{-8cm}
\resizebox{2.2\textwidth}{!}{%
  \includegraphics{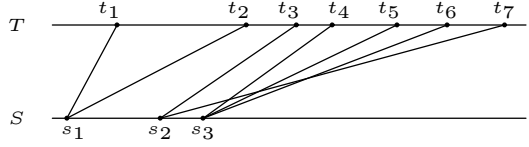}
}
\vspace{-33cm}
\caption{An example for Case B where $t=7$, $s=3$ and $C_S=\{3,2,3\}$.}
\label{fig:14}       
\end{figure}

\end{description}

Note that 

\[Y\left(b_i\right)=Y\left(b_{i-1}\right)+e_{s-i+1}+f_i+\]\[min(C\left(a_{s-i},{\alpha}_{s-i}\right),C\left(a_{s-i+1},{\alpha }_{s-i+1}-1\right))-\]\[min(C\left(a_{s-i+1},{\alpha}_{s-i+1}\right),C\left(a_{s-i+2},{\alpha}_{s-i+2}-1\right)),\] and \[R_{ih}=R_{(i-1)h}+f_i+e_k,\] where $a_k$ is the largest point in $A_w$ that is not saturated in the matching corresponding to $R_{(i-1)h}$. And finally, the values of $X\left(b_i,j\right)$ for all $1\le j\le {\beta }_i$ is computed in $O(s+t)$ time, since \[S'_{i-1}=S'_i+e_i+min\left(C\left(a_{i-2},{\alpha }_{i-2}\right),C\left(a_{i-1},{\alpha }_{i-1}-1\right)\right)-\]\[min\left(C\left(a_{i-1},{\alpha }_{i-1}\right),C\left(a_i,{\alpha}_i-1\right)\right),\] 
and ${S'}_{i,j}={S'}_{i,j-1}+f_j-b_k$, where $b_k$ is the largest point of $A_{w+1}$ in the matching corresponding to ${S'}_{i,j-1}$ that its degree is more than one. So the values of $X\left(b_i,j\right)$, $Y\left(b_i\right)$, and $Z\left(b_i\right)$ can be computed in $O(s+t)$ time. Therefore, we can compute $C\left(b_i,k\right)$ for all $1\le i\le t$ and $1\le k\le {\beta}_i$ in $O(n)$ time, and our algorithm computes an ODMLM-matching between two point sets with total cardinality $n$ in $O(n^2)$ time.
\qquad\end{proof}

 \section{Concluding Remarks}
 \label{Conclsection}

The limited capacity many-to-many point matching is a many-to-many matching where each point has a limited capacity. In this paper, we study the one dimensional limited capacity many-to-many matching problem, called ODMLM-matching problem, in which we match two point sets that lie on the real line. We provide an algorithm that determines an ODMLM-matching between two point sets with total cardinality $n$ in $O(n^2)$ time. The two-dimensional version of this problem is open which may be solved using the geometric information.



\bibliographystyle{elsarticle-num}
\bibliography{<your-bib-database>}



\end{document}